\documentclass[12pt]{article}
\usepackage{amsmath,amssymb,amsfonts,mathtools}
\usepackage{longtable}

\newtheorem{prop}{Proposition}

\def\rem{{\noindent{\bf Remark.}\ }}

\allowdisplaybreaks
\newenvironment{proof}[1][Proof]{\noindent\textbf{#1.} }{\
  \rule{0.5em}{0.5em}}

\def\sds{\strut \displaystyle}

\title{On Solving A Generalized Chinese Remainder Theorem in the Presence of Remainder Errors}
\author{Guangwu Xu\thanks{Department of EE \& CS, University of Wisconsin-Milwaukee,
Milwaukee, WI 53211, USA;
e-mail: {\tt gxu4uwm@uwm.edu}.
Research
supported in part by the National 973 Project of China (No.
2013CB834205).}
}
\date{}

\begin{document}
\maketitle

\begin{abstract} In estimating frequencies given that the signal
waveforms are undersampled multiple times, Xia et. al.
proposed to use a generalized version of Chinese remainder Theorem (CRT),
where the moduli are $M_1, M_2, \cdots, M_k$ which are not necessarily pairwise coprime.
If the errors of the
corrupted remainders are within $\tau=\sds \max_{1\le i\le k} \min_{\stackrel{1\le j\le k}{j\neq i}} \frac{\gcd(M_i,M_j)}4$,
their schemes can be used to
 construct an approximation of the solution to the generalized CRT
with an error smaller than $\tau$. Accurately finding the quotients is
a critical ingredient in their approach.
In this paper, we shall start with a faithful historical account of the generalized CRT.
We then present two
treatments of the problem of solving generalized CRT with erroneous remainders.
 The first treatment follows the route of Wang and Xia to find
the quotients, but with a
simplified process. The second treatment considers a simplified model of generalized CRT and
takes a different approach by working on the
corrupted remainders directly. This approach also reveals some useful information about
the remainders by inspecting extreme values of the erroneous remainders modulo $4\tau$.
Both of our treatments produce efficient algorithms with essentially optimal performance.
Finally, this paper constructs a counterexample to prove the sharpness of the error bound $\tau$.

\end{abstract}


\section{Introduction}
The usual Chinese Remainder Theorem (CRT) concerns  reconstructing an integer given its remainders with respect to
a set of coprime moduli. More precisely, let ${\cal M}=\{m_1, m_2, \dots , m_k\}$ be a set of pairwise
coprime positive integers and $\sds M=\prod_{i=1}^k m_i$. For a set of
integers $r_1, r_2, \dots, r_k$ with $0\le r_i < m_i$, the Chinese
Remainder Theorem says that the system of congruences
\begin{equation}\label{eq:crt}
\left\{ \begin{array}{l} x\equiv r_1 \pmod {m_1} \\
                               x\equiv r_2 \pmod {m_2} \\
                                \cdots \\
                               x\equiv r_k \pmod {m_k}
\end{array} \right.
\end{equation}
has a unique solution $0\le x < M$. In fact, using the extended
Euclidean algorithm (Aryabhattiya algorithm), one finds integers $u_1, u_2, \cdots, u_k$ such
that
\begin{equation}\label{daiyanfor1}
\sum_{i=1}^k u_i \frac{M}{m_i} = 1.
\end{equation}
This implies that for any integer $N$, we have $\sds N=\sum_{i=1}^k N u_i \frac{M}{m_i}$. In particular, if $N\equiv r_i \pmod {m_i}$ for
$i = 1, 2, \cdots, k$, then
\begin{equation}\label{eq:2.1}
N = \sum_{i=1}^k r_i u_i \frac{M}{m_i} \pmod M
\end{equation}
i.e., the right hand side gives us the desired solution. For a greater efficiency, one can also use a randomized algorithm by Cooperman, Feisel, von zur Gathen, and Havasin
in \cite{CFGH} to get the coefficients $\{u_i: i=1,\cdots, k\}$ of (\ref{daiyanfor1}), see \cite{dlx} for a concrete description.

The Chinese Remainder Theorem is a classical tool that is of both theoretical and practical interests.  A solution of the number theoretic
version of CRT was given by Jiushao Qin (aka Ch'in Chiu-shao) \cite{Qin} under the name of
``method of DaiYan aggregation''. The key subroutine in this method is the ``method of DaYan deriving one'' (which is exactly the method of
finding modulo inverse). We would like to remark that the CRT discussed in \cite{Qin} is actually in a general form, namely,
the moduli are not assumed to be pairwise coprime ( and the book did describe a way of reducing the general case to the case of pairwise coprime moduli ),  see also \cite{Ore,struik,XuLi}.
A form that is similar to (\ref{daiyanfor1}) was discussed in
\cite{Qin} as well.
It is also interesting to note that the proof of the ring theoretic version of CRT (which contains the finite discrete
Fourier transform as a special case) uses the similar idea as that from \cite{Qin}.

One of the natural applications of CRT is in
parallel computing, as CRT can be interpreted as an isomorphism between a ring (of bigger size) and a direct product of rings (of smaller
sizes). Since arithmetic operations on those smaller rings are closed and independent, an operation in the bigger ring may be
decomposed into operations in those smaller rings in parallel \cite{dlx,gxu}. Applications of CRT in fast computation can be found
in cryptography \cite{stinson}, e.g., fast decryption in RSA, solving the discrete logarithm with a composite modulus, and
the fast point counting method for elliptic curves . In a series of work by Xia and his collaborators, a
generalized CRT (the so called Robust CRT) has been used
to solve problems of estimating frequencies given that the signal
waveforms are undersampled multiple times, \cite{llx,WangXia,XiaWang}. For example, in \cite{WangXia},  Wang and  Xia discussed the following form of CRT, which
is to construct an approximation to the solution of the system of congruences
\begin{equation}\label{simpleVer}
\left\{ \begin{array}{l} x\equiv r_1 \pmod {dm_1} \\
                               x\equiv r_2 \pmod {dm_2} \\
                                \cdots \\
                               x\equiv r_k \pmod {dm_k}
\end{array} \right.
\end{equation}
given that the remainders are corrupted. They design a nice scheme to solve the problem by first finding the quotients $q_i=\frac{x-r_i}{m_i}$.
They proved that if the remainder errors are bounded by $\frac{d}4$, the quotients can be computed exactly.
Recently, Xiao, Xia and Wang \cite{XiaoXiaWang} proposed a new robust CRT that accommodates more general situation where
the great common divisors of the pairs of moduli $\{M_1, M_2, \cdots, M_k\}$ are not necessarily the same. In the case, the
error bound is changed to $\tau=\sds \max_{1\le i\le k} \min_{\stackrel{1\le j\le k}{j\neq i}} \frac{\gcd(M_i,M_j)}4$.
We note that the multi-stage use of robust CRT proposed in \cite{XiaoXiaWang} provides an interesting way of
lifting some restriction on error bound.

In this paper, we shall start with a faithful interpretation of the work of Qin  \cite{Qin} concerning CRT. In our discussion
of ``method of DaYan aggregation'', we will describe an exact translation of
Qin's ``method of DaYan deriving one'' into using modern pseudo-code \cite{XuLi}. Compared with those appeared in
literature, our description reflected Qin's original idea. The second part of the paper
studies the problem of recovering solution of generalized CRT under the presence of noisy remainders
by giving alternative solutions and more efficient computational procedures.
We shall mainly work with the simple model (\ref{simpleVer}), as the general case can be dealt with essentially by
one of the approaches in our discussion. We present two treatments to solve this problem. The first treatment follows the route of Wang and Xia to  find
the quotients, but with a simplified process. The second treatment takes a different approach by working on the
corrupted remainders directly. This approach also reveals some useful information about
the remainders by inspecting extreme values of the erroneous remainders modulo $d$.
The technical difficulties mentioned in \cite{WangXia}
have been successfully avoided by the method proposed in our second
treatment. Both of our treatments produce efficient algorithms that invoke CRT only once,
and they are essentially optimal in this sense. We remark that our first treatment can handle the
case with more general moduli.

It is a natural question to ask whether the
remainder error bound $\frac{d}4$ is sharp. We provide an affirmative answer to this question
by constructing a simple counterexample.

The rest of our discussion will be divided into three sections. After introducing  Qin's ``DaYan aggregation'' method
with some discussions and comments in Section 2, the
Section 3 describes the problem setup and provides solutions. Section 4 is
the conclusion.

\section{On Qin's methods of  ``DaYan Aggregation'' and ``DaYan Deriving One''}
In his book ``The Mathematical Treatise in Nine Sections'' \cite{Qin} (published in 1247),  Qin
described ``DaYan aggregation'' method which is a generalized version of
the Chinese Remainder Theorem. This method is aimed at solving system of congruence equations of the following form
\begin{equation}\label{genCRT0}
\left\{ \begin{array}{l} x\equiv  r_1  \pmod {M_1} \\
                               x\equiv  r_2   \pmod {M_2} \\
                                \cdots \\
                               x\equiv  r_k  \pmod {M_k}
\end{array} \right.
\end{equation}
without assuming that the moduli are pairwise coprime, but requiring $\gcd(M_i,M_j)| r_i-r_j$. Qin's outlined an efficient
 procedure that reduces (\ref{genCRT0}) to the
following form (the usual Chinese remainder representation):
\begin{equation}\label{CRT}
\left\{ \begin{array}{l} x\equiv r_1 \pmod {c_1} \\
                               x\equiv r_2 \pmod {c_2} \\
                                \cdots \\
                               x\equiv r_k \pmod {c_k}
\end{array} \right.
\end{equation}
with a suitable set $\{c_1, c_2,\cdots, c_k\}$ of pairwise coprime positive integers such that $c_i|M_i$ and
$\mbox{lcm}(M_1, M_2,\cdots,M_k)=c_1c_1\cdots c_k$.

The main technical
component of ``DaYan aggregation'' is ``DaYan deriving one'', which is Qin's method of computing modulo inverse:
given positive integers $m> a>1$, compute $a^{-1} \pmod m$.
The ``DaYan deriving one'' states (see also \cite{Libbrect} \footnote{The interpretation from \cite{Libbrect} (page 331) contains some error.}):

{\bf DaYan deriving one}:
\begin{itemize}
\item {\tt Set up the number $a$ at the right hand above, the number $m$ at the right hand below.
Set  $1$ at the left hand above.}

\item {\tt First divide the `right below' by the `right above', and the quotient\\
obtained, multiply it by the $1$ of `left above' and add it to `left below'}.

\item {\tt  After this, in the `upper' and `lower' of the right column,
divide larger number by smaller one. Transmit  and divide them by each other.
Next bring over the quotient obtained and [cross-] multiply with each other. Add the `upper' and the `lower' of the
left column.}

\item {\tt One has to go until the last remainder of the `right above' is $1$ and then one can stop.
Then you examine the result of `left above'; take it as the modulo inverse.}
\end{itemize}

This ancient procedure is very close to a modern pseudo-code. It keeps a state of four variables in a form
of $2\times 2$ matrix:  $ \begin{pmatrix}\mbox{\tt left-above} & \mbox{\tt righ-above} \\\mbox{\tt left-below} & \mbox{\tt right-below} \end{pmatrix} \triangleq\begin{pmatrix}x_{11}&x_{12}\\x_{21}&x_{22}\end{pmatrix} $ and the initial state is $\begin{pmatrix}1 & a \\0 & m \end{pmatrix}$. Then
the procedure executes steps which are exactly a while-loop. The termination condition of the while-loop is ``{\tt until the last remainder of the `right above' is $1$}''.

This termination condition seems not to have been interpreted correctly (see several papers in \cite{Wu_ed}). In literature, the usual
(positive) integer division
\[
  c = \bigg\lfloor \frac{c}d\bigg\rfloor d + r
  \]
is used and the remainder $r$ is the least nonnegative residue modulo $b$, i.e., $0\le r <b$. With this division,
it may happen that {\tt `right below'} becomes  $1$ first which makes ``{\tt until the last remainder of the `right above' is $1$}'' (i.e. $x_{21}=1$) to be
unachievable. The existing research papers from \cite{Wu_ed} and the book \cite{Libbrect}  suggested that one should modify Qin's procedure  when  $x_{22}=1$ is seen, in order to
make $x_{21}=1$. However, we believe that Qin made no mistake in his termination condition, namely, after an even number of steps (this is another interesting
fact of Qin's design), $x_{21}=1$ can always be achieved. In \cite{XuLi}, a detailed explanation has been given.
We shall make a brief account here: in ancient China, another form of division could be used.
Such a setting takes the remainder to be the least positive residue modulo the divisor. As an example, a divination method using ``I Ching'' (Book of Change, 1000-400 BC)
is to generate a hexagram by the manipulation of $50$ yarrow stalks. In this process, division by $4$ is used and the remainder
must belong to $\{1,2,3,4\}$. It should be noted that Qin also described this divination method in his book \cite{Qin}. This division can be expressed as:
for positive integers $c$ and $d$, there is a unique $r$ with $1\le r\le d$, such that
\[
  c = \bigg\lfloor \frac{c-1}d\bigg\rfloor d + r.
\]
The remainder $r$ is the least positive residue modulo  $d$. This sort of division is also mentioned  in \cite{Wang}.

This kind of division enables us to formulate Qin's algorithm in modern language which is faithful to his original idea; in particular,
$x_{21}=1$ can always be achieved \cite{XuLi}.
\begin{center}
\begin{tabular}{|l|}
\hline
${}\quad $ {\bf Qin's Algorithm: DaYan Deriving One} \\
 \hline
{\bf Input:} $\quad a, m$ with $1<a<m, \gcd(a, m)=1$,\\
{\bf Output:} positive integer $u$ such $ua \equiv 1 \pmod m$.\\
 \hline
 \hskip 0.2cm $\begin{pmatrix}x_{11}&x_{12}\\x_{21}&x_{22}\end{pmatrix} \gets \begin{pmatrix}1 & a \\0 & m \end{pmatrix}$;\\
 \hskip 0.2cm {\tt while } ($x_{12} > 1 $) {\tt do }\\
 \hskip 0.5cm {\tt if ( $x_{22} > x_{12}$ )}\\
 \hskip 0.5cm${}\quad\quad  q \gets \lfloor \frac{x_{22}-1}{x_{12}} \rfloor$ ; \\
 \hskip 0.5cm${}\quad\quad  r \gets x_{22} - q x_{12};$\\
 \hskip 0.5cm${}\quad\quad  x_{21} \gets q x_{11}+x_{21};$\\
 \hskip 0.5cm${}\quad\quad  x_{22} \gets r;$\\
 \hskip 0.5cm {\tt if ( $x_{12} > x_{22}$ )}\\
\hskip 0.5cm ${}\quad\quad  q \gets \lfloor \frac{x_{12}-1}{x_{22}} \rfloor$ ; \\
\hskip 0.5cm ${}\quad\quad  r \gets x_{12} - q x_{22};$\\
\hskip 0.5cm${}\quad\quad  x_{11} \gets q x_{21}+x_{11};$\\
 \hskip 0.5cm${}\quad\quad  x_{12} \gets r;$\\
 \hskip 0.2cm $u\gets x_{11}$;\\
 \hline
 \end{tabular}
\end{center}
We have the following remarks on this formulation of Qin's algorithm.
\begin{enumerate}
\item Notice that when $x_{22} > x_{12}$, the variable $x_{12}$ will not be updated. As we start with $m>a$, the value of $x_{12}$
 can only be updated at an even numbered  step. In particular, we see a very interesting fact: $x_{12}=1$ will be reached in an even number of steps.
\item If at an odd numbered step, $x_{22}=1$, then $x_{12} > x_{22}$ must hold. In this case $q  = x_{12} - 1, r = 1$. So
the next step gives $x_{12}  = 1$.
\item Set $v = -\frac{ua-1}m$, we get the B\'ezout identity:  $ua+vm =1$.
\item Another important observation of  Qin's algorithm is that the permanent of the state $\quad\quad$ $\begin{pmatrix}x_{11}&x_{12}\\x_{21}&x_{22}\end{pmatrix} $
is an invariant:
\[
x_{11}x_{22}+x_{12}x_{21} = m.
\]
\end{enumerate}
Qin used his method of ``DaYan deriving one'' to solve the generalized CRT (\ref{genCRT0}). We would like to
remark that Qin paid a special attention to the following  relation
\begin{equation}\label{eq:fanyong}
\sum_{i=1}^k \bigg(\bigg(\frac{M}{c_i}\bigg)^{-1}\pmod {c_i}\bigg) \frac{M}{c_i} = 1+ gM,
\end{equation}
where $M=c_1\cdots c_k$ and $g$ is a positive integer. This also leads to a direct formula of CRT as
that we  did in (\ref{eq:2.1}).

\section{Problem Setup and Solutions}

Our discussion will be mainly for the simple model (\ref{simpleVer}). Namely, we will discuss solving the problem of following generalized CRT
\begin{equation}\label{eq:gencrt}
\left\{ \begin{array}{l} x\equiv r_1 \pmod {dm_1} \\
                               x\equiv r_2 \pmod {dm_2} \\
                                \cdots \\
                               x\equiv r_k \pmod {dm_k}
\end{array} \right.
\end{equation}
in the presence of remainder errors. We will
give two different solutions to the problem, and one of them can be easily adopted to dealing with the general case.

It is clear that when $d>1$, the existence of a solution $0\le x <dM$ is equivalent to $r_i-r_j \equiv 0\pmod d \quad \mbox{ for all } i, j.$
Moreover, as mentioned in \cite{WangXia}£¬ the  solution  $N$ of (\ref{simpleVer}) satisfies
\begin{equation}\label{eq:gensol}
N=dN_0+ a
\end{equation}
where $a = r_1 \pmod d $ and $0\le N_0<M$ is the solution of the CRT system
\[
\left\{ \begin{array}{l} y\equiv  \gamma_1  \pmod {m_1} \\
                               y\equiv  \gamma_2   \pmod {m_2} \\
                                \cdots \\
                               y\equiv  \gamma_k  \pmod {m_k}
\end{array} \right.
\]
where $\gamma_i=\frac{r_i-a}d, \ i=1,2,\cdots, k$.

As discussed at the beginning, the solution of (\ref{eq:gencrt}) can  also be written as
\[
N = \sum_{i=1}^k r_i u_i \frac{M}{m_i} \pmod {dM},
\]
where $\sds M=\prod_{i=1}^k m_i$ and  $\sds \sum_{i=1}^k u_i \frac{M}{m_i}=1$ is the relation (\ref{daiyanfor1}).
However, the solution of the form (\ref{eq:gensol}) is of certain interest because it reveals the fact that the parameter $d$
introduces redundancy.  Small changes of $r_i$'s will not
affect the integral part of $\frac{r_i-a}d$ for most cases. It is noted that $a$ is another important parameter, and prior
knowledge or estimation of this number may be useful in getting a better approximation of $N$.
 We shall call $a= r_1 \pmod d $\footnote{Throughout this paper,
the expression $ g = h \pmod m$ means that $g$ is the least nonnegative remainder of $h$ modulo $m$, i.e.,
$m|(g-h)$ and $0\le g < m$.} the {\sl common remainder} modulo $d$.

Given a corrupt set of remainders
$\overline{r}_1,\cdots,\overline{r}_k$\footnote{We assume that each of these observed remainder satisfies $0\le \overline{r}_i<dm_i$. Otherwise, errors can be reduced by
simply setting $\overline{r}_i=dm_i-1 \ ( \mbox{ or } \overline{r}_i=0 )$ if $\overline{r}_i\ge dm_i \ (  \mbox{ or } \overline{r}_i<0 )$. } with errors $\Delta r_i =\overline{r}_i-r_i$, our
task is to find a good approximation $\overline{N}$ to the solution $N$ of (\ref{eq:gencrt}).
We shall assume $|\Delta r_i|<\frac{d}4$ for all $i=1,2,\cdots, k$.

\subsection{Reconstruction by finding the quotients}
A nice observation by Wang and Xia \cite{WangXia} is that if  $|\Delta r_i|<\frac{d}4$ for all $i=1, \cdots, k$, then one is able to reconstruct
an approximation $\overline{N}$ of the solution $N$ such that
\[
|N-\overline{N}| <\frac{d}4.
\]
The interesting strategy used in \cite{WangXia} is to consider the quotients $\sds q_i= \left\lfloor \frac{N}{m_i}\right\rfloor =\frac{N-r_i}{m_i}$ and
prove that they are invariants when the remainder errors are bounded by $\frac{d}4$. The procedure in \cite{WangXia}
is summarized as

\begin{longtable}{|ll|}
\caption{\bf Algorithm of \cite{WangXia}}\\
\hline
\endfirsthead
\multicolumn{2}{l}%
{\tablename\ \thetable\ -- \textit{Continued from previous page}} \\
\hline
\endhead
\hline
\multicolumn{2}{l}{\textit{Continued on next page}} \\
\endfoot
\hline
\endlastfoot
\hline
{\bf Step 1}.& For each $i=2,\cdots,k$ \\
&${}\quad$ find $\Gamma_{i,1} = m_i^{-1} \pmod {m_1}$; \\
&${}\quad$ compute $\overline{\xi}_{i,1}=\left[\frac{\overline{r}_i-\overline{r}_1}d\right]\Gamma_{i,1}\pmod {m_i}$.\\
{\bf Step 2}.& Use CRT to find the solution $\overline{q}_1$: \\
&${}\quad \left\{ \begin{array}{l} \overline{q}_1\equiv  \overline{\xi}_{2,1}  \pmod {m_2} \\
                                \cdots \\
                               \overline{q}_1\equiv  \overline{\xi}_{k,1}  \pmod {m_k}
\end{array} \right.$\\
{\bf Step 3}.& For each $i=2,\cdots,k$ \\
&${}\quad \sds \overline{q}_i = \frac{\overline{q}_1 m_i - \left[\frac{\overline{r}_i-\overline{r}_1}d\right]}{m_i}$.\\
{\bf Step 4}.& For each $i=1,2,\cdots,k$ \\
&${}\quad$ compute $\overline{N}(i) = dm_i + \overline{r}_i$; \\
&Compute $\overline{N} = \left[\frac{1}k \sum_{i=1}^k \overline{N}(i)\right] $.\\
\hline
\end{longtable}
Here $[x]$ denotes the rounding integer, i.e., $x-\frac{1}2 \le [x]<x+\frac{1}2$.

\rem  Under the conditions $|\Delta r_i|<\frac{d}4$,
one has $\left[\frac{\overline{r}_i-\overline{r}_1}d\right]=\frac{r_i-r_1}d$ and hence $\overline{q}_i =q_i$.
It is also remarked that the computation in step 1
is equivalent to an invocation of CRT. Thus, the above algorithm uses two invocations of solving CRT.

The main purpose of this subsection is to present a concise approach to this problem. For each $1\le j \le k$, recall that the quotient  $q_j$ is
$q_j = \frac{N-r_j}{dm_j}.$
Denote $M_j=\frac{M}{m_j}$. From (\ref{daiyanfor1}) ( i.e., $ \sum_{i=1}^k u_i \frac{M}{m_i} = 1$ ), we have
\begin{eqnarray}\label{eq:quotient}
q_j  &=&  q_j u_1 \frac{M}{m_1} +\cdots + q_j u_k \frac{M}{m_k} \equiv \sum_{\stackrel{i=1}{ i\neq j}}^k\frac{r_i-r_{j}}{d} u_i \frac{M}{m_im_j}\pmod{M}\notag\\
& = & \sum_{\stackrel{i=1}{ i\neq j}}^k\frac{r_i-r_{j}}{d} u_i \frac{M_j}{m_i}\pmod{M_j}.
\end{eqnarray}
Note that $q_j$ can be uniquely specified because $0\le q_j<M_j$.

Since  $\left[\frac{\overline{r}_i-\overline{r}_1}d\right]=\frac{r_i-r_1}d$ for $i=2,\cdots,k$, the following is true by (\ref{eq:quotient})
\[
q_1=\overline{q}_1=\sum_{i=2}^k\left[\frac{\overline{r}_i-\overline{r}_{1}}{d}\right] u_i \frac{M_j}{m_i}\pmod{M_1}.
\]
Therefore, we get the following procedure by replacing steps 1 and 2 of the algorithm in \cite{WangXia} with the above discussion.
\begin{longtable}{|ll|}
\caption{\bf Algorithm 1}\\
\hline
\endfirsthead
\multicolumn{2}{l}%
{\tablename\ \thetable\ -- \textit{Continued from previous page}} \\
\hline
\endhead
\hline
\multicolumn{2}{l}{\textit{Continued on next page}} \\
\endfoot
\hline
\endlastfoot
\hline
{\bf Step 1}.& Use the extended Euclidean algorithm to get $u_i$ such that \\
&${}\quad\sds \sum_{i=1}^k u_i \frac{M}{m_i} = 1 $. \\
{\bf Step 2}.& Compute  $\overline{q}_1$: \\
&${}\quad \sds\overline{q}_1=\sum_{i=2}^k\left[\frac{\overline{r}_i-\overline{r}_{1}}{d}\right] u_i \frac{M_1}{m_i}\pmod{M_1}$\\
{\bf Step 3}.& For each $i=2,\cdots,k$ \\
&${}\quad \sds \overline{q}_i = \frac{\overline{q}_1 m_1 - \left[\frac{\overline{r}_i-\overline{r}_1}d\right]}{m_i}$.\\
{\bf Step 4}.& For each $i=1,2,\cdots,k$ \\
&${}\quad$ compute $\overline{N}(i) = \overline{q}_i dm_i + \overline{r}_i$; \\
&Compute $\overline{N} = \left[\frac{1}k \sum_{i=1}^k \overline{N}(i)\right] $.\\
\hline
\end{longtable}

\rem 1. Steps 1 and 2 are equivalent to an invocation of CRT. In step 3, $q_i$ may be computed via (\ref{eq:quotient}). One computation of CRT
is saved compared to the corresponding algorithm of \cite{WangXia}.

\rem 2. Recently Xiao, Xia, and Wang \cite{XiaoXiaWang} extended the problem to approximate the solution $N$ of the general Chinese remainder system
(\ref{genCRT0}):
\[
\left\{ \begin{array}{l} x\equiv  r_1  \pmod {M_1} \\
                               x\equiv  r_2   \pmod {M_2} \\
                                \cdots \\
                               x\equiv  r_k  \pmod {M_k}
\end{array} \right.
\]
from erroneous remainders $\overline{r}_1,\cdots,\overline{r}_k$. Where $d_{ij}=\gcd(M_i,M_j)$ can be arbitrary, and the remainder
errors are bounded by $\tau=\sds \max_{1\le i\le k} \min_{\stackrel{1\le j\le k}{j\neq i}} \frac{d_{ij}}4$. Notice that
 (\ref{genCRT0}) has a solution if and only if $d_{ij}|r_i-r_j$ for all $i,j$.
We remark that the
procedure we just described can be adopted to this problem without much difficulty. In fact, assuming $\tau= \min_{2\le j\le k} \frac{d_{1j}}4$ without
loss of generality, we can recover the quotient $q_1=\frac{N-r_1}{M_1}$. To this end, we let $M=\mbox{lcm}\{M_1,\cdots, M_k\}$. Since
$\gcd(\frac{M}{M_1},\cdots, \frac{M}{M_k})=1$, we get integers $v_1,\cdots, v_k$ such that $\sds\sum_{i=1}^k v_i \frac{M}{M_i}=1$. This yields
$\sds\sum_{i=1}^k q_1 v_i \frac{M}{M_i}=q_1$. A routine manipulation gives
\[
q_1 = \frac{r_2-r_1}{d_{12}} v_2\frac{d_{12}M}{M_1M_2}+\cdots+\frac{r_k-r_1}{d_{12}} v_k\frac{d_{1k}M}{M_1M_k}\pmod {\frac{M}{M_1}}.
\]
Since $\frac{r_j-r_1}{d_{1j}}=\left[\frac{\overline{r}_j-\overline{r}_1}{d_{1j}}\right]$ for $j=2,\cdots,k$, $q_1$ can be computed exactly from
the corrupted remainders.

\subsection{Extreme values of the erroneous remainders modulo $d$}
In this subsection, we will take a different approach to solve the problem under the simple model (\ref{eq:gencrt}).
We will use corrupted remainders directly in solving CRT. The
main points we would like to make include
\begin{enumerate}
\item The extreme values such as $\max\{\overline{r}_i \pmod d \}$ and $\min\{\overline{r}_i \pmod d \}$ should be inspected to reveal useful information
about the errors.
\item The common remainder $a$ maybe shifted so that more accurate estimation can be made.
\item The ideas developed in this subsection naturally lead to a proof of the sharpness of the error bound.
\end{enumerate}
 Let $a=r_1 \pmod d$ be the common remainder of $r_1, \cdots, r_k$ modulo $d$.
We define
\begin{eqnarray*}
&& \alpha = \max\{\overline{r}_i \pmod d : i=1,2,\cdots, k\}\\
&& \beta = \min\{\overline{r}_i \pmod d : i=1,2,\cdots, k\}\\
&& \mu=\min\{ \overline{r}_i \pmod d: \overline{r}_i \pmod d >\frac{d}2\}\\
&& \nu=\max\{ \overline{r}_i \pmod d: \overline{r}_i \pmod d <\frac{d}2\}\\
\end{eqnarray*}
The numbers $\mu, \nu$ are defined only when the corresponding sets are nonempty.

Since  $|\Delta r_i|<\frac{d}4$ for all $ i= 1,2,\cdots, k$, we have the following five mutually exclusive cases: \vspace{-7mm}
\begin{description}
\item{\bf (a)} For all $ i= 1,2,\cdots, k$, $a+\Delta r_i < 0$.\\
In this case, $\alpha-\beta <\frac{d}4$. We also have $a<\frac{d}4$.
\item{\bf (b)} For all $ i= 1,2,\cdots, k$, $a+\Delta r_i \ge d$. \\
In this case, $\alpha-\beta <\frac{d}4$.
\item{\bf (c)} For all $ i= 1,2,\cdots, k$, $0\le a+\Delta r_i < d$. \\
In this case, $\alpha-\beta < \frac{d}2$.
\item{\bf (d)} There are $i_1$ and $j_1$ such that $a+\Delta r_{i_1} < 0$ and $a+\Delta r_{j_1}\ge 0$. \\
In this case, we must have $a<\frac{d}4$, $\alpha-\beta \ge \mu-\nu > \frac{d}2$
\item{\bf (e)}  There are $i_1$ and $j_1$ such that $a+\Delta r_{i_1} < d$ and $a+\Delta r_{j_1}\ge d$. \\
In this case, $\alpha-\beta \ge \mu-\nu > \frac{d}2$.  We also have $a>\frac{3d}4$
\end{description}

We note that even though these five cases cannot be checked individually (due to the unknown parameters $\Delta r_i$ and $a$), we can still
divide them into two verifiable situations. In fact, it is easy to see that
the condition $\alpha-\beta < \frac{d}2$ covers  cases {\bf (a), (b)} and {\bf (c)}, while the condition
$\alpha-\beta > \frac{d}2$ covers  cases {\bf (d)} and {\bf (e)}.

We shall first get  good approximations of $\gamma_i=\frac{r_i-a}d$
based on these conditions.

\begin{prop}
For $i=1,2, \cdots, k$, set
\[
\overline{\gamma}_i =\left\{ \begin{array}{ll} \left\lfloor \frac{\overline{r}_i}{d}\right\rfloor & \mbox{ if } \alpha-\beta < \frac{d}2,\\
                             \left\lfloor \frac{\overline{r}_i+d-\mu }{d}\right\rfloor & \mbox{ if } \alpha-\beta \ge \frac{d}2.
                             \end{array}\right.
\]
Then we have \vspace{-7mm}
\begin{enumerate}
\item  For case {\bf (a)}, $\overline{\gamma}_i = \gamma_i -1$ holds for all $i=1,2,\cdots, k$.
\item  For cases {\bf (c)} and {\bf (d)}, $\overline{\gamma}_i = \gamma_i $ holds for all $i=1,2,\cdots, k$.
\item  For cases {\bf (b)} and {\bf (e)}, $\overline{\gamma}_i = \gamma_i +1$ holds for all $i=1,2,\cdots, k$.
\end{enumerate}
\end{prop}
\begin{proof}
The verification is fairly straightforward. We just consider the cases {\bf (b)} and {\bf (e)} here. We recall that $r_i=\gamma_id+a$.

For case {\bf (b)}, since $a+ \Delta r_i\ge d$ for all $i$, so
\[
\overline{\gamma}_i=\left\lfloor \frac{\overline{r}_i}{d}\right\rfloor=\left\lfloor \frac{\gamma_i d+a+\Delta r_i}{d}\right\rfloor=\gamma_i +1.
\]

For case {\bf (e)}, we know that if $\overline{r}_i \pmod d<\frac{d}2$, then $a+\Delta r_i\ge d$. This implies that
\[
d\le a+\Delta r_i+d-\mu < 2d.
\]
If $\overline{r}_j \pmod d\ge \frac{d}2$, then $a+\Delta r_j< d$ and
$\overline{r}_j \pmod d=a+\Delta r_j$. This, together with the minimality of $\mu$, implies that
\[
d\le a+\Delta r_j+d-\mu < 2d.
\]
Therefore,
\[
\overline{\gamma}_i=\left\lfloor \frac{\overline{r}_i+d-\mu }{d}\right\rfloor=\left\lfloor \frac{\gamma_i d+a+\Delta r_i+d-\mu}{d}\right\rfloor=\gamma_i +1.
\]
\end{proof}

Now let us present an algorithm to get an approximation $\overline{N}$ of $N$, based on the above discussion.
\small{
\begin{longtable}{|ll|}
\caption{\bf Algorithm 2}\\
\hline
\endfirsthead
\multicolumn{2}{l}%
{\tablename\ \thetable\ -- \textit{Continued from previous page}} \\
\hline
\endhead
\hline
\multicolumn{2}{l}{\textit{Continued on next page}} \\
\endfoot
\hline
\endlastfoot
\hline
{\bf Step 1}.& Compute $\alpha, \beta$ and $\mu$ / / (if $\mu$ exists ) \\
{\bf Step 2}.& For each $i=2,\cdots,k$, compute\\
     & $\sds\overline{\gamma}_i =\left\{ \begin{array}{ll} \left\lfloor \frac{\overline{r}_i}{d}\right\rfloor & \mbox{ if } \alpha-\beta < \frac{d}2,\\
                             \left\lfloor \frac{\overline{r}_i+d-\mu }{d}\right\rfloor & \mbox{ if } \alpha-\beta \ge \frac{d}2.
                             \end{array}\right.$\\
{\bf Step 3}.& Use CRT to find the solution $\overline{N}_0$: \\
             &  ${}\quad \left\{ \begin{array}{l} \overline{N}_0\equiv  \overline{\gamma}_1 \pmod {m_1} \\
                                  \overline{N}_0\equiv  \overline{\gamma}_2 \pmod {m_2} \\
                                \cdots \\
                               \overline{N}_0\equiv  \overline{\gamma}_k \pmod {m_k} \\
                              \end{array} \right.$\\
{\bf Step 4}.& If $\alpha-\beta < \frac{d}2$\\
             & ${}\quad\quad \overline{N} = d \overline{N}_0 + \left\lfloor \frac{\sum_{i=1}^k (\overline{r}_i\pmod d )} {k}\right\rfloor$\\
             & Else / / (i.e., $\alpha-\beta \ge \frac{d}2$) \\
             & ${}\quad\quad \overline{N} = d \overline{N}_0$.\\
\hline
\end{longtable}
}\normalsize{}
Again, this is an efficient way of finding an approximation of $N$ as
only one CRT computation is required. What left to show is that this approximation is indeed a good one.
\begin{prop} Let $N$ be the solution of (\ref{eq:gencrt}) and $\overline{N}$ be the output of Algorithm 2, we have
\[
|N-\overline{N}| <\frac{d}4.
\]
\end{prop}
\begin{proof} Let $N_0$ be the solution of (\ref{eq:crt}), then we know that $N=dN_0+a$.

For case {\bf (a)}. In this case, $\overline{N}_0=N_0-1$.
Since all $a+\Delta r_i<0$, we must have $\overline{r}_i \pmod d = d+a+\Delta r_i$ and $-\frac{d}4<a+\Delta r_i<0$. Thus
\begin{eqnarray*}
\overline{N} &=& d \overline{N}_0 + \left\lfloor \frac{\sum_{i=1}^k (\overline{r}_i\pmod d )} {k}\right\rfloor
=dN_0-d + \left\lfloor \frac{\sum_{i=1}^k (d+a+\Delta r_i)} {k}\right\rfloor\\
&=&dN_0-d+\left\lfloor d+a+\frac{\sum_{i=1}^k \Delta r_i} {k}\right\rfloor=N+\left\lfloor \frac{\sum_{i=1}^k \Delta r_i} {k}\right\rfloor.
\end{eqnarray*}
For case {\bf (b)}. In this case, $\overline{N}_0=N_0+1$. Since all $a+\Delta r_i\ge 0$, we must have $\overline{r}_i \pmod d = a+\Delta r_i -d $ and $0\le a+\Delta r_i-d <\frac{d}4$. Thus
\begin{eqnarray*}
\overline{N} &=& d \overline{N}_0 + \left\lfloor \frac{\sum_{i=1}^k (\overline{r}_i\pmod d )} {k}\right\rfloor
=dN_0+d + \left\lfloor \frac{\sum_{i=1}^k (a+\Delta r_i-d)} {k}\right\rfloor\\
&=&dN_0+d+\left\lfloor a-d+\frac{\sum_{i=1}^k \Delta r_i} {k}\right\rfloor=N+\left\lfloor \frac{\sum_{i=1}^k \Delta r_i} {k}\right\rfloor.
\end{eqnarray*}
For case {\bf (c)}. In this case, $\overline{N}_0=N_0$. We know that in this case $\overline{r}_i \pmod d=a+\Delta r_i$, so
\begin{eqnarray*}
\overline{N} &=& d \overline{N}_0 + \left\lfloor \frac{\sum_{i=1}^k (\overline{r}_i\pmod d )} {k}\right\rfloor
=dN_0 + \left\lfloor \frac{\sum_{i=1}^k (a+\Delta r_i)} {k}\right\rfloor\\
&=&dN_0+\left\lfloor a+\frac{\sum_{i=1}^k \Delta r_i} {k}\right\rfloor=N+\left\lfloor \frac{\sum_{i=1}^k \Delta r_i} {k}\right\rfloor
\end{eqnarray*}
So in these three cases, $|N-\overline{N}|<\frac{d}4$ holds.

For case {\bf (d)}. In this case, $\overline{N}_0=N_0$, and $\overline{N} =dN_0$. We know that in this case $a <\frac{d}4$, so
\[
|N-\overline{N}| =a <\frac{d}4.
\]

For case {\bf (e)}. In this case, $\overline{N}_0=N_0+1$ , and $\overline{N} =dN_0+d$. We know that in this case $a >\frac{3d}4$, so
\[
|N-\overline{N}| =d-a <\frac{d}4.
\]
\end{proof}

\subsection{ The sharpness of the condition $\max_{1\le i\le k} |\Delta r_i| \le \frac{d}4$}.
In this subsection, we shall discuss some issues about the error bound.
It is easy to see that if we are given $a=\left\lfloor\frac{d}2\right\rfloor$, then the case where the error bound is as big as
$\left\lfloor\frac{d}2\right\rfloor$ can be handled. However, we have no prior knowledge about $a$ in general.  To the best of our knowledge,
 $\frac{d}4$ is the largest error bound available in literature.  It is thus interesting to ask whether the bound $\frac{d}4$ can be improved.
Recall that in our algorithm 2, we need to deal with the cases $\alpha-\beta <\frac{d}2$ and $\alpha-\beta >\frac{d}2$ differently. This
is very suggestive. Pushing this consideration further, we see that if an error is at least $\frac{d}4$ in magnitude, then such distinction
is no longer available. In fact, we are able to show that the bound $\frac{d}4$ cannot be improved in general (i.e., it
is a sharp bound) by the following simple counterexample.

{\bf Example}. Let $p, q$ be distinct primes, and $d$ be a positive integer divisible by $4$.
We consider solving the system of congruence equations
\begin{equation}\label{eq:exa}
\left\{ \begin{array}{l} x\equiv  r_1 \pmod {d p} \\
                                 x\equiv  r_2 \pmod {d q} \\
                              \end{array} \right.
\end{equation}
in the presence of remainder errors.

Assume that the remainder errors are allowed to be $|\Delta r_1|\le \frac{d}4$ and $|\Delta r_2|\le \frac{d}4$
(these are equivalent to $|\Delta r_i|< \frac{d}4+1$).
We will show that for some corrupted remainders, it is impossible determine an approximation of the true solution.
To this end, suppose we have corrupted remainders
\[
\overline{r}_1 = d, \ \overline{r}_2 = \frac{3d}2.
\]

First, we consider the system (\ref{eq:exa}) with $r_1 = \frac{3d}4, \ r_2=\frac{7d}4=d+\frac{3d}4$ (i.e., in this case $a=\frac{3d}4$).
We get the corresponding solution
\[
N_1=d vp + \frac{3d}4
\]
where $v = p^{-1} \pmod q$.
Since $\overline{r}_1 = r_1+\frac{d}4, \ \overline{r}_2 = r_2 - \frac{d}4$, the erroneous remainders $\overline{r}_1 ,  \overline{r}_2$ are legitimate
for this case.

Next, we consider the system (\ref{eq:exa}) with $r_1 = r_2=\frac{5d}4=d+\frac{d}4$ (i.e., in this case $a=\frac{d}4$). We get the corresponding solution
\[
N_2=\frac{5d}4=d+\frac{d}4.
\]
Since $\overline{r}_1 = r_1-\frac{d}4, \ \overline{r}_2 = r_2 + \frac{d}4$, the erroneous remainders
$\overline{r}_1,  \overline{r}_2$ are legitimate
for this case too.

If we choose $p$ to be large, then $N_1, N_2$ are far apart. Therefore, no approximation
based on $\overline{r}_1 ,  \overline{r}_2$ can be close to both of them. In other words, if
the error bound is bigger than $\frac{d}4$, the problem of
solving  (\ref{eq:exa}) with erroneous remainders is not identifiable in general.
We also note that the equality $\alpha-\beta=\frac{d}2$ holds in both cases.

\section{Conclusion}
The generalized Chinese Remainder Theorem has been used by Xia et. al. to model some signal processing problems.
In this paper we fist present a faithful historical account of the work of Jiushao Qin concerning generalized CRT.
Efficient procedures of constructing an approximate solution for the generalized CRT (simplified vmodel) based on
corrupted remainders are proposed in this paper. These two procedures improve that of Xia et. al. in that we only
need one computation of CRT. They are asymptotically optimal since at least one CRT is required. Our first procedure can
be adopted to the general CRT model. The ideas in our treatment of the second procedure might be of some independent interest.
We also provide a proof of the sharpness for the error bound $\frac{d}4$.

\section*{Acknowledgement} The author would like to thank Drs. Ian Blake and Kumar Murty for their
constructive comments.

\end{document}